\documentclass[conference,a4paper]{IEEEtran}
\usepackage{epsfig}
\usepackage{times}
\usepackage{float}
\usepackage{afterpage}
\usepackage{amsmath}
\usepackage{mathrsfs}
\usepackage{amstext}
\usepackage{amssymb,bm}
\usepackage{latexsym}
\usepackage{thmtools, thm-restate}
\usepackage{color}
\usepackage{thm-restate}
\usepackage{graphicx}
\usepackage{amsmath}
\usepackage{amsthm}
\usepackage{tikz}
\usepackage{tikzscale}
\usepackage{graphicx}
\usepackage[center]{caption}
\usepackage{pstricks}
\usepackage{subcaption}
\usepackage{booktabs}
\usepackage{multicol}
\usepackage{lipsum}% dummy text
\usepackage{dblfloatfix}
\usepackage{algorithm} 
\usepackage[noend]{algpseudocode} 
\usepackage[normalem]{ulem}
\usepackage{enumerate}
\usepackage{bbm}
\usepackage{dsfont}
\usepackage{tikz}
\usepackage{thmtools}

\newcommand{\mcal}{\mathcal}

\newtheorem{thm}{Theorem}%[section]

\newtheorem*{defin*}{Definition}

\algnewcommand\algorithmicforeach{\textbf{for each}}

\algdef{S}[FOR]{ForEach}[1]{\algorithmicforeach\ #1\ \algorithmicdo}

\begin{document}

\title{Secure Communication over \\ 1-2-1 Networks} 
\author{
\IEEEauthorblockN{Gaurav Kumar Agarwal$^\star$, Yahya H. Ezzeldin$^\star$, Martina Cardone$^\dagger$, Christina Fragouli$^\star$ }
$^\star$ University of California Los Angeles, Los Angeles, CA 90095, USA\\
Email: \{gauravagarwal, yahya.ezzeldin, christina.fragouli\}@ucla.edu \\
$^\dagger$ University of Minnesota, Minneapolis, MN 55404, USA, Email: cardo089@umn.edu
%\thanks{The work of the authors was partially funded by NSF under
%award 1321120. 
%G. K. Agarwal is also supported by the Guru Krupa Fellowship.}
}
\IEEEoverridecommandlockouts

\maketitle

\begin{abstract}
%Authors in~\cite{yahya}, examined the capacity of wireless networks that employ mmWave communication.
%A model referred to as Gaussian 1-2-1 network was formulated, and min-cut bounds were derived along with achievable strategies. 
This paper starts by assuming a 1-2-1 network, the abstracted noiseless model of mmWave networks that was shown to closely approximate the Gaussian capacity in~\cite{yahya},  and studies secure communication.
%In this paper, *** we do not do the abstraction in this paper ****
%{the Gaussian model is abstracted} into noiseless edges of fixed finite capacities and the secure communication over such noiseless 1-2-1 networks {is analyzed.}
First, the secure capacity is derived for 1-2-1 networks where a source is connected to a destination through a network of unit capacity links. 
Then, lower and upper bounds on the secure capacity are derived for the case when source and destination have more than one beam, which allow them to transmit and receive in multiple directions at a time. 
Finally, secure capacity results are presented for diamond 1-2-1 networks when edges have different capacities.
  %we study a special network called  the diamond network and characterize its capacity when all edges can have different capacities which are not necessary the same. 
\end{abstract}

\section{Introduction}
High-frequency communication, such as mmWave and Thz, can enable multi-gigabit communication, albeit at relatively short range, and with the help of beamforming to compensate for high path loss. To cover large areas, such as commercial buildings, requires deploying networks of relays that communicate through directional beams. In~\cite{yahya}, the authors derived a model for high-frequency communication networks, that they termed Gaussian 1-2-1 networks, and presented capacity results as well as information flow algorithms. {In this paper, we start by assuming a 1-2-1 network, namely the abstracted noiseless model of mmWave networks that was shown to closely approximate the Gaussian capacity in \cite{yahya} and, study secure message communication over such networks.}
%
%In this paper, we abstract {from} these Gaussian 1-2-1 networks into noiseless 1-2-1 networks with edges of fixed finite capacities, and focus on secure message communication over such networks.

The 1-2-1 model abstracts directivity:  to establish a communication link, both the mmWave transmitter and receiver employ antenna arrays that they electronically steer to direct their beams towards each other - termed as 1-2-1 link, as both nodes need to focus their beams to face each other for the link to be active. Thus, in 1-2-1 networks, instead of broadcasting or interference, we have coordinated steering of transmit and receive beams to activate different links at each time. 

%Our main result in \cite{yahya} is that, over such networks with one source, one destination and $N$ relays,  we can approximately (i.e., up to a gap that only depends on $N$) achieve the capacity\footnote{With a slight abuse of terminology, in the rest of the paper we refer to this constant gap capacity approximation simply as capacity.} by time-sharing  routing of information across paths; moreover, out of an exponential number {(in $N$)} of paths that potentially connect the source to the destination, we  {showed} we need to utilize at most a linear number  {(in $N$)} of them to achieve the capacity. For example, in a diamond network, where the source is connected to the destination through one layer of non-interfering {relays,} if the links have unit capacity, we only need to route information along one of the $N$ possible paths;  if the links have arbitrary capacity, we need to route information along at most two paths~\cite{yahya}. 

%    \begin{figure}
%        \centering
%        \includegraphics[width=0.48\textwidth]{beams.png}
%        \caption{Diamond Network with Beam pointing to some directions}
%        \label{fig:example}        
%    \end{figure}    

    We now review a fundamental result in network security, and then discuss how it changes over 1-2-1 networks.  Consider a source, Alice, connected to a destination, Bob, through an arbitrary traditional network represented as a graph with unit capacity lossless links, and assume that the min-cut between the source and the destination equals $H$. That is, we can find $H$ edge-disjoint unit capacity paths that connect the source to the destination. Assume that a passive eavesdropper, Eve, wiretaps any $K$ links of the communication network. Alice can then securely (in the strong information theoretical sense) communicate at rate $H-K$ with the destination, by conveying linear combinations of  $K$ keys with $H-K$ information messages \cite{cai2002secure}. The rate $H-K$ is exactly the secure message capacity\footnote{This holds under some standard assumptions in the literature \cite{cai2002secure}, in particular under the assumption that only Alice can generate randomness.} - we cannot hope to do better.
    %For example, if the mincut is $h=2$, and if Eve eavesdrops any one edge, we can securely send information at rate 1.
 
    In 1-2-1 networks of unit capacity edges, it turns out that even if the 1-2-1 {min-cut} is $H$, i.e., the maximum flow using mmWave communication is $H$, and Eve eavesdrops any $K$ edges, we may be able to securely communicate at rates higher than $H-K$.  Consider for example a diamond network with $N$ relays shown in Fig.~\ref{fig:diamond_unequal} with all edges of unit {capacity:} the unsecure communication capacity equals one - we cannot do better than rate one because Alice can beamform and transmit information at only one relay at each time, and it does not matter which relay she communicates with, since we assumed that all links have unit capacity.
    Assume that  Eve wiretaps any one edge. That is, we have $H=1$ and $K=1$, which over traditional networks would result to a zero secure communication rate.  However, Alice can vary which relay she communicates with over time; in fact, she can devote a fraction $\frac{1}{N}$ of her time 
 to send information to Bob over any one out of the $N$ unit capacity paths that connect them. Because Eve will only be observing one of these paths, as we formally show in Section~\ref{sec:unit_edge}, 
{Alice can securely communicate at rate} of $1-\frac{1}{N}$. This is closer to the unsecure communication rate of one, than to zero. That is,  {for security} over 1-2-1 networks, {we can leverage} the fact that we may have many possible choices of paths to achieve the unsecure capacity, to {communicate at} rates much higher than $H-K$.

\noindent{\bf{Main Contributions.}}
(a) We consider arbitrary 1-2-1 networks with unit capacity edges, {where} Eve wiretaps any $K$ edges, and derive lower and upper bounds on the capacity, that are tight for some networks.
(b) We derive the secure message capacity for the case where the source is connected to the destination through {one layer of non-interfering relays (i.e., diamond network)}, where now each path from the source to the destination can have arbitrary capacity.

\noindent{\bf{Related Work.}}
In our work, we essentially leverage directivity and multipath for security, over a ``lossless'' network model. 
The fact that directivity can help with security has been observed {in the context of} MIMO beamforming, see~\cite{mimo_beamform} and later {work~\cite{iris};} {in these works,} the main observation is that, by creating a narrow beam, we limit the locations where the adversary Eve can collect useful information - or at least, significantly weaken her channel, so as to utilize wiretapping coding. 
{However, to the best of our knowledge,}
%As far as we know, 
these ideas have not been extended to networks. Exploiting multipath for security {over} lossless networks with unit {capacity} links has notably been used in secure network {coding~\cite{cai2002secure}.} This was followed by a number of works such as~\cite{khaleghi2009subspace,langberg,jaggi2005polynomial,Agarwal2017}.  For edges with {non-uniform} capacities, Cui et al.~\cite{cui2013secure} designed a secure achievable scheme. 
  These results however, consider only the ``traditional network'', where a node can communicate to other nodes using all the edges it is connected with, and not 1-2-1 networks where a node with one beam can only transmit to one among its neighbors at each point in time.
  
  \noindent{\bf{Paper Organization.}} 
Section~\ref{sec:model} presents  the {1-2-1} network model and results on the unsecure capacity for arbitrary networks with unit edge capacities. Section~\ref{sec:unit_edge} contains secure capacity results for arbitrary networks with unit edge capacities and Section~\ref{sec:diamond} presents our secure capacity results for diamond networks with arbitrary edge capacities. 
{Section~\ref{sec:Concl} concludes the paper.}
%  
%Section~\ref{sec:model} presents  the $N$-relay Gaussian {1-2-1} network model  and reviews its approximate capacity in~\cite{yahya}.  Section~\ref{sec:unit_edge} contains results for arbitrary networks with unit edge capacities and Section~\ref{sec:diamond} presents our results for diamond networks with arbitrary edge capacities.

%\section{System Model and Unsecure Capacity}
%\label{sec:model}
%\textcolor{blue}{This section needs significant work}.

\section{System Model {And} Unsecure Capacity}
\label{sec:model}
%In this section, we first review the 1-2-1 model introduced in~\cite{yahya} as well as capacity characterizations for unsecure communication, and then describe the adversary model and the security metrics used throughout the paper.
\noindent\textbf{Notation.} $[m]: = \{1, \ 2, \ \ldots, \ m \}$, $[a : b]$ is the set of integers from $a$ to $b \geq a$ {and}
$A_{[m]} = \{A_1, \ A_2, \ A_3, \ \ldots, \ A_m\}$.

\smallskip

\noindent
    {\bf 1-2-1 Model.}
{The work in \cite{yahya} examined capacity characterizations for unsecure communication, and showed that we can approximately achieve the capacity of a Gaussian mmWave network within a constant gap, by considering instead of the underlying Gaussian network, a lossless network that was termed  1-2-1 network model. In this paper, we examine security over such 1-2-1 networks, that we  describe next.}
We consider a source connected to {a destination} through a directed acyclic graph  $G = (V,E)$  with edges of fixed finite capacities, {where} each link can be activated according to the 1-2-1 constraints. That is, at any particular time, an intermediate node can simultaneously receive and transmit but it can at most listen to one node (one incoming edge) and direct its transmission to one node (one outgoing edge) in the network. 
The source (respectively, destination) can transmit to ({respectively,} receive from) $M$ other nodes i.e., on $M$ outgoing edges ({respectively,} on $M$ incoming edges), simultaneously with no interference.
%{\color{red}In the case where we consider unit capacity edges, we will assume that the 1-2-1 min-cut also equals $M$, i.e., there exist $M$ vertex disjoint paths connecting the source to the destination.}

\smallskip

\noindent
{\bf Adversary Model and Security.}
We assume that the source wishes to communicate a message $W$ of entropy rate $R$ securely from a passive {external} adversary Eve who can wiretap any $K$ edges of her choice. 
%{For our analysis, we will assume that all edge capacities and the entropy rate of the source message $W$ are over a finite field $\mathbb{F}_q$ which is sufficiently large. MC: I am not sure about this assumption. Our constant gap characterization in~\cite{yahya} is for Gaussian noise networks (i.e., the characterization in~\eqref{eq:constGap} might not hold if the network is not Gaussian). Do we really need this assumption?} 

If Eve wiretaps edges in the set $S \subseteq E, \ |S| = K$, and  the symbols transmitted on these edges over $n$ network uses are denoted by $\{T^n_{e}, \ e \in S \}$, then we require that:
\begin{align}
I(W; \{T^n_{e}, \ e \in S \} ) & \leq \epsilon, \ \forall S \subseteq E, \ |S| = K \label{eq:security} .
\end{align}
We are interested in {characterizing} the secure message capacity $C$, using the standard {definition} of the maximum rate at which the source can communicate with the destination under~\eqref{eq:security}.

\smallskip

\noindent \textbf{Unsecure Capacity:} 
{Here, we derive the capacity {in the absence of the eavesdropper Eve.} 	1-2-1 {networks} with arbitrary edge capacities, and $M=1$, under Gaussian channel models are analyzed in~\cite{yahya}, where {the main result} is that {over such networks, one can} \textbf{approximately} (i.e., up to a gap that only depends on $N$) achieve {the capacity} {by routing information} across paths; moreover, out of an exponential number {(in $N$)} of paths that potentially connect the source to the destination, capacity can be achieved by utilizing at most a linear number (in $N$) of them. In this section, we derive an additional results, namely the exact capacity for any $M$ when all the edges are of unit capacity.}

	\begin{thm}
		For arbitrary 1-2-1 networks with unit capacity edges, the capacity in absence of {Eve} is given by,
		\begin{align}
		C_u = \min(M, H_v),
		\end{align}where $H_v$ is the maximum number of vertex disjoint paths in the network.
	\end{thm}
	\begin{IEEEproof}
	\textbf{Achievability:}
		Let {$p_{[H_v]}$} be {the} $H_v$ vertex disjoint {paths.} The fact that paths are vertex disjoint is {crucial under the 1-2-1 constraints. This is because intermediate nodes} can transmit and receive from only one node each, and this ensures that multiple paths (depending on {the} number of source and destination beams) can be {simultaneously} operated at each time. We pick $\min(M, H_v)$ such paths and use {these} for the transmission and thus achieve a rate of $\min(M,H_v)$.		
	\\
	\textbf{Outer Bound:} {Whenever} there are direct edges from the source to the destination, we add a virtual node in between, so that {a} direct edge {turns} into a {two-hop} path. 
This does not change {the transmission rate} as if there was a transmission {on the direct edge in $G$,} it can also be {performed} using the added {virtual} node with no extra resources.  Thus, we can assume that there are no direct edges from the source to the destination.
	
Now, {we consider} the minimum vertex cut of the network, i.e., the minimum number of vertices (excluding {the} source and {the} destination), {such that when we remove them} there is no path from the source {to} the destination. 
This minimum number of vertices {is} equal to the maximum number of vertex disjoint paths, i.e., $H_v$. {We denote} these vertices as $V_1, \ V_2, \ , \ldots, \ V_{H_v} $. {Each} of these intermediate nodes can transmit only on one of {its} outgoing edges. {We} denote the symbols transmitted on the outgoing edges of these nodes over $n$ network {uses as} $T^n_{V_{[H_v]}}$, where $T^n_{V_i}$ denotes the symbols transmitted by vertex $V_i$. {We represent the} symbols received by the destination {as} $T^n_{D}$.
	
	By Fano's inequality, {we obtain}
	\begin{align*}
	n R & \leq H(W) \stackrel{(a)}{=} H(W)- H(W|T^n_{D}) \\
	& \stackrel{(b)}{\leq} H(W) - H(W|T^n_{V_{[H_v]}}) \\
	& = I(W; T^n_{V_{[H_v]}}) \leq H(T^n_{V_{[H_v]}}) \stackrel{(c)}{\leq} n H_v. \\
	nR & \leq H(W)- H(W|T^n_{D})  = I (W; T^n_{D}) \\
	& \leq H(T^n_D) \stackrel{(d)}{\leq} M n. \\
	R & \leq \min(M, H_v), 
	\end{align*}
	where $(a)$ is due to {the} reliable decoding {constraint;} $(b)$ {follows from the `conditioning does not increase the entropy' principle and since} $V_{[H_v]}$ {is} a vertex cut and thus all {the} information going to the destination passes through {these vertices (i.e., $T^n_{D}$ is a deterministic function of $T^n_{V_{[H_v]}}$);} $(c)$ is because there are $H_v$ symbols for every instance and there are $n$ such instances; and $(d)$ holds because the destination can receive only on $M$ incoming edges from $M$ nodes.	     	
	\end{IEEEproof}

\section{Arbitrary Networks with Unit Link Capacity}
\label{sec:unit_edge}
{In this section, we}
%We next 
prove lower and upper bounds on the secure capacity.
%To derive these bounds, we consider two different types of min-cut: the 1-2-1 {min-cut,} that takes into account the 1-2-1 constraints, 
% such that each relay can only transmit and receive from {at most} one neighbor at each point in time; and the min-cut associated with the  underlying graph $G$ if we were to ignore the 1-2-1 constraints - if this {was} a regular {digraph} connecting the source to the destination. 

\begin{thm}\label{thm:1}
	Consider an arbitrary 1-2-1 network with unit capacity edges.
%	{\color{red}, and 1-2-1 min-cut equal to $M$ between the source and the destination}
	\begin{itemize}
		\item[(a)] For $M=1$: If $H_e$ is the maximum number of \textbf{edge disjoint} paths connecting the source to the destination on the underlying graph, then the 1-2-1 secure capacity $C$ can be lower bounded as follows:
		\begin{align}
		C \geq  \left(1-  \frac{K}{H_e} \right).
\label{eq:lbM1}
		\end{align}
		\item[(b)] For $M>1$: If $H_v$ is the maximum number of \textbf{vertex disjoint} paths connecting the source to the destination on the underlying graph, then the 1-2-1 secure capacity $C$ can be lower bounded as follows:
		\begin{align}
		C \geq \min(M, H_v) \left(1-  \frac{K}{H_v} \right).
\label{eq:lbMgreat1}
		\end{align}
	\end{itemize}
%        \textcolor{red}{Since we have defined $M$ to be the min-cut in the 1-2-1 network, it follows directly that $H_v$ is $\geq M$ - there definitely exist $M$ vertex disjoint paths - do we need the min above?}
\end{thm}

\begin{proof}
  The main intuition {behind} the proof is that we can apply the optimal secure communication scheme we would have used on the underlying graph if we did not have the 1-2-1 constraints, and then {use} this scheme under the 1-2-1 constraints, as described in {what follows.} \\
  %	\begin{itemize}
  (a) \textbf{For $M=1$}:
		Let {$p_{[H_e]}$}
%$p_1, \ p_2, \ \ldots, \ p_{H_e}$
be the edge disjoint paths. 
We start by {generating} $K$ random packets and make $H_e$ linear {combinations} of these using an MDS code matrix of size $K \times H_e$. 
{We refer to} these packets as {$X_{[H_e]}$.} Any $K$ of these combinations are mutually independent. Next, we take $H_e - K$ message packets, and add {(i.e., encode)} these with {the first}
%first 
$H_e- K$ {random packets.}
%of extended random packets.
{In other words, after this coding operation we obtain}
%Let's say our message packets are $W_1,\ W_2, \ \ldots, \ W_{H_e-K}$. After adding (encoding) with key packets, we get $T_1, \ T_2, \ \ldots, \ T_{H_e}$, where 
		\begin{align*}
		T_i & = \left\{ \begin{array}{cl}
		W_i + X_i  &  \text{if } i \leq H_e - K \\
		X_i & \text{else} \end{array} \right. ,	
		\end{align*}
{where $W_{[H_e-K]}$ are message packets.}
		
		We use the network $H_e$ times, and in each instance we use one of the {paths} from {$p_{[H_e]}$.} Thus, we would be able to communicate all encoded symbols in $H_e$ time instances. {Moreover,} the destination will be {able to} cancel out the keys and thereby decode $H_e - K$ messages, as there are $K$ symbols {$T_{[H_e-K+1:T_{H_e}]}$,} 
%$(T_{H_e}, \ T_{H_e - 1}, \ \ldots, \ T_{H_e-K+1}  )$ 
which are just {independent combinations of the} $K$ random packets we started with.
		
		Moreover, in each instance, Eve will receive a symbol if {the edges she eavesdrops} are part of the path {that is used in that particular instance.}
Since her $K$ edges can at most be part of $K$ paths, {Eve} will receive {at most} $K$ symbols, all of {which} are encoded with independent keys. Thus, the scheme securely {transmits} $H_e- K$ message packets in $H_e$ uses of the network. Hence, we get {a rate $R=\frac{H_e - K}{H_e}=1-\frac{K}{H_e}$, which is precisely the one in~\eqref{eq:lbM1}.} 
Note that security follows from the security of the underlying scheme, that is a standard scheme for multipath security.\\
{\bf (b) For $M>1$: }	Let {$p_{[H_v]}$} be the vertex disjoint {paths.} Again, the fact that paths are vertex disjoint is {crucial under the 1-2-1 constraints. This is because intermediate nodes} can transmit and receive from only one node each, and this ensures that $M$ paths can be {simultaneously} operated at each time (note {that} having vertex disjoint paths is a sufficient but not a necessary condition).

 Let $\hat{M} = \min(M,H_v)$. We start by {generating} $K {{H_v - 1}\choose{\hat{M} - 1}}$ {random} packets and extend them to $ \hat{M} {H_v  \choose\hat{M}}$ packets using an MDS code matrix. 
Then, {similar to the case $M=1$, we take the first $\hat{M} {H_v\choose\hat{M}} - K {{H_v - 1}\choose{\hat{M} - 1}}$ of these random packets and add (i.e., encode) them with the same amount of message packets.}
%we take  $\hat{M} {H_v\choose\hat{M}} - K {{H_v - 1}\choose{\hat{M} - 1}} $  message packets and add first these many random extended packets to these messages as we did previously. 
More formally, if $\left \{X_i, i \! \in \! \left[\hat{M} {H_v  \choose\hat{M}}\right]   \right \}$ are the random packets after the extension using {the} MDS code matrix, {and} $\left \{ W_i,  i \!\in\! \left[ \hat{M} {H_v\choose\hat{M}} \!-\! K {{H_v - 1}\choose{\hat{M} - 1}}\right] \right \}$ are the message packets, then
		\begin{align*}
		T_i & = \left\{ \begin{array}{cl} X_i + W_i & \text{if } i \leq \hat{M} {H_v\choose\hat{M}} - K {{H_v - 1}\choose{\hat{M} - 1}}  \\ X_i & \text{else} \end{array} \right. .
		\end{align*}	
		{We use} this network ${H_v \choose \hat{M} }$ times, and in each instance we use a different choice of $\hat{M}$ paths to communicate. 
{It is not difficult to see that each of the $K$  edges {eavesdropped by the} adversary will intersect with ${H_v - 1 \choose \hat{M} - 1}$ {such network uses.} This is because, for a fixed choice of edge, there are ${H_v - 1 \choose \hat{M} - 1}$ network instances where {a symbol} in carried via this edge.} Hence, in total {the} adversary will receive only  $K {H_v - 1 \choose \hat{M} - 1}$ symbols, which {are} encoded with independent keys.
The receiver, after the ${H_v \choose \hat{M} }$ {network uses} will be able to cancel out the keys. Thus, we can securely communicate $ \hat{M} {H_v  \choose \hat{M}} - K {H_v - 1 \choose \hat{M} - 1} $ over  ${H_v \choose \hat{M} }$ instances of the network, and  achieve {a rate $R$ equal to}
		\begin{align*}
		R & = \frac{\hat{M} {H_v  \choose \hat{M}} - K {H_v - 1 \choose \hat{M} - 1}}{{H_v \choose \hat{M}} }\\
		& =  \hat{M} - \frac{K \hat{M}}{H_v} \\
		& = \min(M, H_v) \left(1 -  \frac{K}{H_v} \right),
		\end{align*} 		
{which is precisely the one in~\eqref{eq:lbMgreat1}. This concludes the proof of Theorem~\ref{thm:1}.}
\end{proof}

\begin{thm} \label{thm:2}
	Let $H_e$ be the maximum number of \textbf{edge disjoint} paths connecting the source to the destination on the underlying directed graph, then the 1-2-1 secure capacity $C$ can be upper bounded as {follows:}
	\begin{align*}
	C \leq \min(M, H_e ) \left(1-  \frac{K}{H_e} \right).
	\end{align*}
\end{thm}

\begin{proof}
	From the min-cut, max-flow theorem there {are} $H_e$ edges such {that, when removed, the} source gets disconnected from the destination. 
Let $e_1, \ e_2, \ \ldots, \ e_{H_e}$ {denote} these edges. 
{Assume that} the network is used $n$ times, and let $T^n_{e_i}, \ i \in \{1,2,\ldots, H_e\}$ be the symbols transmitted on these $H_e$ edges over $n$ uses {of the} network. 
	{By denoting} the symbols transmitted by the source on $n$ network instances by $T^n_S$, then,	
	\begin{align*}
	n M & \geq  H(T^n_S) \stackrel{(a)}{=} H(T^n_S, \{T^n_{e_i}, \ i \in [H_e] \}) \\
	& \geq H(\{T^n_{e_i}, \ i \in [H_e] \}),
	\end{align*}
	where $(a)$ follows because $\{T^n_{e_i}, \ i \in [H_e] \}$ is a deterministic function of $T^n_S$. {Moreover,} $H(\{T^n_{e_i}, \ i \in [H_e] \})  \leq n H_e$. Thus, 
\begin{align}
	H(\{T^n_{e_i}, \ i \in [H_e] \}) & \leq \min(nH_e, nM). \label{eq:min_entropy}
	\end{align}	
	
{In the remaining part of the proof, we use the result in the following lemma, which is proved}
%	The following lemma used in this proof is proved 
in the Appendix. 
	
		\begin{restatable}{lem}{infotheory}
			\label{lem:infotheory}
			$\forall m,$ there {exists} a set $S \subset [L], |S| = m $, such that  $H(\{X_i, \ i \in S^c\} | \{X_i, \ i \in S\})  \leq \frac{L - m}{L} H (\{X_i, \ i \in [L]\})$.
		\end{restatable}

{Without loss of generality,} for $m = K$, {we} assume {$S = [K] \subset [H_e]$} in Lemma~\ref{lem:infotheory}. Then, by Fano's inequality, {we have}
	\begin{align*}
	n R & \leq H(W) = H(W) - H(W|\{T^n_{e_i}, \ i \in [H_e]  \}) \\
	& = I(W ; \{T^n_{e_i}, \ i \in [H_e]  \} ) \\
	&  = I (W; \{T^n_{e_i}, \ i \in [K]  \} ) + \\
	& \quad I (W; \{T^n_{e_i}, \ i \in [H_e]\setminus[K] \} | \{T^n_{e_i}, \ i \in [K]  \} )	\\ 
	&  \stackrel{(a)}{\leq} \epsilon + I (W; \{T^n_{e_i}, \ i \in [H_e]\setminus[K] \} | \{T^n_{e_i}, \ i \in [K]  \} )	\\
	& \leq \epsilon + H (\{T^n_{e_i}, \ i \in [H_e]\setminus[K] \} | \{T^n_{e_i}, \ i \in [K]  \} ) \\
	& \stackrel{(b)}{\leq}  \epsilon + \frac{H_e - K}{H_e}\min(nH_e, nM) \\
	& {\implies} R  \leq  \min(M, H_e) \left(1 - \frac{K}{H_e}\right),
	\end{align*}
	where $(a)$ {follows since, for security,} $I (W; \{T^n_{e_i}, \ i \in [K]  \} ) \leq \epsilon$ and $(b)$ is because of Lemma~\ref{lem:infotheory} and~\eqref{eq:min_entropy}. {This concludes the proof of Theorem~\ref{thm:2}.}
\end{proof}

\subsection{Discussion}
{For} some special cases, we can exactly characterize the capacity {(i.e., the upper and lower bounds previously derived coincide).} In particular, {these include:}
\begin{itemize}
\item {Networks} where the number of  edge disjoint paths is equal to the number of vertex disjoint paths. {For these networks,} the capacity if given by $C \!=\! \min(M,H_e) (1- \frac{K}{H_e})$.
  \item For {networks} where the source and the destination {have one} transmit and one receive beam each, i.e., $M=1$. {For these networks, the capacity is given by $C = 1- \frac{K}{H_e}$. }
%  If the 1-2-1 mincut equals one, then the secure capacity is given by $C = 1-\frac{K}{H_e}$, as it is easy to see that we can achieve the upper bound by using one path at a time. 
\end{itemize}
We next provide {two} different network {examples} where: 1) the upper bound is tight {(Example~1) and 2) the outer bound is not tight, but the lower bound is tight (Example~2).}
\paragraph*{Example 1} 
In Fig.~\ref{fig:tight_arbitrary}, there are four edge disjoint paths from {the} source to the destination, {i.e., $H_e =4$.} Assume that $M=2$, {i.e.,} both the source and the destination can transmit and receive from two nodes {and $K=1$, i.e., Eve wiretaps any one edge of her choice.} 
From Fig.~\ref{fig:tight_arbitrary}, {we refer to} these four paths as $p_1, \ p_2, \ p_3 \ $ and $p_4$, {ordered} from top to bottom. 
To achieve {the outer bound,} one can first use $p_1$ and $p_4$ and then use $p_2$ and $p_3$ to {communicate} two symbols in each instance of network use. Thus, on two time instances, one can communicate $4$ {messages} ($3$ {securely since $K=1$}). This gives a secure rate of $\frac{3}{2}$, which matches the outer bound.

\begin{figure*}
	\centering
	\begin{subfigure}[b]{0.33\linewidth}
		\centering
			\begin{tikzpicture}
			\begin{scope}[scale=0.6]
			\draw (0,0) circle(3mm) node {S};
			
			\draw (2,1.5) circle(3mm);
			\draw (2,0.5) circle(3mm);
			\draw (2,-0.5) circle(3mm);
			\draw (2,-1.5) circle(3mm);
			
			\path[->] (0.3,0) edge (1.7,1.5);
			\path[->] (0.3,0) edge (1.7,0.5);
			\path[->] (0.3,0) edge (1.7,-0.5);
			\path[->] (0.3,0) edge (1.7,-1.5);
			
			\draw (4,-0.5) circle(3mm);
			\draw (4,0.5) circle(3mm);
			
			\path[->] (2.3,1.5) edge (3.7,0.5);
			\path[->] (2.3,0.5) edge (3.7,0.5);
			\path[->] (2.3,-0.5) edge (3.7,-0.5);
			\path[->] (2.3,-1.5) edge (3.7,-0.5);
			
			\draw (6,1.5) circle(3mm);
			\draw (6,0.5) circle(3mm);
			\draw (6,-0.5) circle(3mm);
			\draw (6,-1.5) circle(3mm);
			
			\path[->] (4.3,0.5) edge (5.7,1.5);
			\path[->] (4.3,0.5) edge (5.7,0.5);
			\path[->] (4.3,-0.5) edge (5.7,-0.5);
			\path[->] (4.3,-0.5) edge (5.7,-1.5);
			
			\draw (8,0) circle(3mm) node {D};
			\path[->] (6.3,1.5) edge (7.7,0);
			\path[->] (6.3,0.5) edge (7.7,0);
			\path[->] (6.3,-0.5) edge (7.7,0);
			\path[->] (6.3,-1.5) edge (7.7,0);
			\end{scope}
			\end{tikzpicture}
			\caption{~}
			\label{fig:tight_arbitrary}
	\end{subfigure}%	
	\begin{subfigure}[b]{0.33\textwidth}
		\centering
		\begin{tikzpicture}
		\begin{scope}[scale=0.6]

		\draw (0,0) circle(3mm) node {S};
		
		\draw (2,1.5) circle(3mm);
		\draw (2,0.5) circle(3mm);
		\draw (2,-0.5) circle(3mm);
		\draw (2,-1.5) circle(3mm);
		
		\path[->] (0.3,0) edge (1.7,1.5);
		\path[->] (0.3,0) edge (1.7,0.5);
		\path[->] (0.3,0) edge (1.7,-0.5);
		\path[->] (0.3,0) edge (1.7,-1.5);
		
		\draw (4,-0.5) circle(3mm) node{2};
		\draw (4,0.5) circle(3mm) node {1};
		
		\path[->] (2.3,1.5) edge (3.7,0.5);
		\path[->] (2.3,0.5) edge (3.7,0.5);
		\path[->] (2.3,-0.5) edge (3.7,0.5);
		\path[->] (2.3,-1.5) edge (3.7,-0.5);
		
		\draw (6,1.5) circle(3mm);
		\draw (6,0.5) circle(3mm);
		\draw (6,-0.5) circle(3mm);
		\draw (6,-1.5) circle(3mm);
		
		\path[->] (4.3,0.5) edge (5.7,1.5);
		\path[->] (4.3,0.5) edge (5.7,0.5);
		\path[->] (4.3,0.5) edge (5.7,-0.5);
		\path[->] (4.3,-0.5) edge (5.7,-1.5);
		
		\draw (8,0) circle(3mm) node {D};
		\path[->] (6.3,1.5) edge (7.7,0);
		\path[->] (6.3,0.5) edge (7.7,0);
		\path[->] (6.3,-0.5) edge (7.7,0);
		\path[->] (6.3,-1.5) edge (7.7,0);
		\end{scope}
		\end{tikzpicture}
		\caption{~}
		\label{fig:loose_arbitrary}
	\end{subfigure}%
	\begin{subfigure}[b]{0.3\textwidth}
		\centering
		\begin{tikzpicture}
		\begin{scope}[scale=0.6]
		\draw (0,0) circle(3mm) node {S};
		
		\draw (3,1.5) circle(3mm);
		\draw (3,0.5) circle(3mm);
		\draw (3,-0.25) node{$\vdots$};
		\draw (3,-1.5) circle(3mm);
		
		\path[->] (0.3,0) edge node[pos=0.8, above] {$C_1$} (2.7,1.5) ;
		\path[->] (0.3,0) edge node[pos=0.8, above] {$C_2$} (2.7,0.5);
		\path[->] (0.3,0) edge node[pos=0.8, above] {$C_N$} (2.7,-1.5);
		
		\draw (6,0) circle(3mm) node {D};
		\path[->] (3.3,1.5) edge node[pos=0.2, above] {$C_1$} (5.7,0);
		\path[->] (3.3,0.5) edge node[pos=0.2, above] {$C_2$} (5.7,0);
		\path[->] (3.3,-1.5) edge node[pos=0.2, above] {$C_N$} (5.7,0);	
		\end{scope}
				
		\end{tikzpicture}
		\caption{~}
		\label{fig:diamond_unequal}
	\end{subfigure}%
	\caption{(a) {Network} example $H_e = 4$ {for which} the outer bound is tight for $M=2$. (b) {Network} example with $H_e = 4$ {for which the} outer bound is not tight for $M=2$.(c) Diamond network with {non-uniform} path capacities.}
	
\end{figure*}
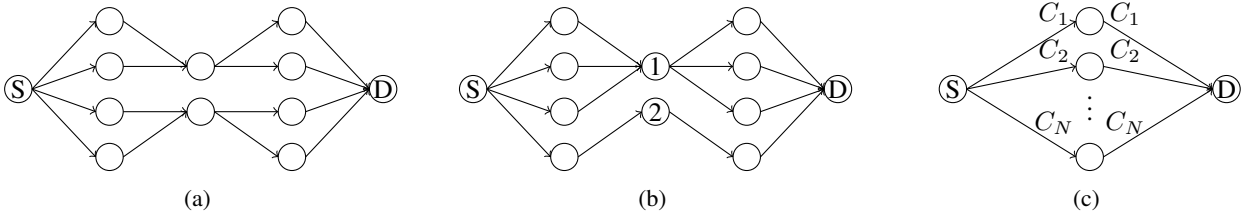

%The following example illustrates a case where the lower bound is tight. 
\paragraph*{Example 2}
Fig.~\ref{fig:loose_arbitrary} {has also $H_e=4$. However, for $M=2$ and $K=1$,} it can be shown that {the} secure capacity is  $1$, whereas {our outer bound in Theorem~\ref{thm:2}} is still $\frac{3}{2}$. 
{In order to achieve a secure rate of one,} we can select {the two paths on the top and on the bottom,} which are node disjoint, and use them to communicate. {We next derive an outer bound for the network in Fig.~\ref{fig:loose_arbitrary} that is tighter than the one in Theorem~\ref{thm:2}.
Assume that, at any time instant $t$, node~$1$ transmits}
%For a tighter bound let at any time instant $t$, node with number 1 in Fig.~\ref{fig:loose_arbitrary}, 
symbol $X_1^{(t)}$ (it can transmit only one symbol even though it has three outgoing edges) {and node~2, transmits} $X_2^{(t)}$. 
Suppose the network is used $n$ times, then by Fano's inequality,
\begin{align*}
n R  \leq H(W) &= H(W) - H\left(W | \{X_i^{(t)}, \ i \in [2], \ t \in [n] \}   \right) \\
& =  I(W; \{X_i^{(t)}, \ i \in [2], \ t \in [n] \}) \\
& =  I (W; \{X_2^{(t)}, \ t \in [n] \}) + \\
& \quad I(W; \{X_1^{(t)}, \ t \in [n] \} | \{X_2^{(t)}, \ t \in [n] \} ) \\
& \stackrel{(a)}{\leq}  \epsilon  + n \implies R  \leq 1,
\end{align*}
where $(a)$ is because, if {Eve wiretaps the edge outgoing} from node $2$, {then} $I (W; \{X_2^{(t)}, \ t \in [n] \}) \leq \epsilon$ and there are only $n$ symbols in $\{X_1^{(t)}, \ t \in [n] \}$.

\section{Diamond Networks with Non-Uniform Path Capacities}
\label{sec:diamond}
For the {$N$-relay} diamond network (shown in Fig.~\ref{fig:diamond_unequal}) with unit edge {capacities,} the lower and upper bounds in Theorem~\ref{thm:1} and {Theorem~\ref{thm:2}} match (since all the $N$ edge disjoint paths are {also} vertex {disjoint,} namely $H_e = H_v=N$), and thus the secure capacity equals  $C = \min(M, N)(1- \frac{K}{N})$.

We next consider the case where the edges have {non-uniform capacities.} In particular, we assume that path {$i \in [N]$} that connects the source to the destination through relay $i$ has capacity $C_i$, as depicted in Fig.~\ref{fig:diamond_unequal}.  
In general, even over traditional networks, the problem of security over unequal capacity edges {is everything but easily solvable~\cite{cui2013secure}. The main reason is that}
%one reason being that 
we need to consider all possible subsets {of edges that Eve may wiretap.}
%, a property that also shows up on 1-2-1 networks.

\begin{thm}\label{thm:3}
	For the diamond network with $M=1$ and $N$ relays as shown in Fig.~\ref{fig:diamond_unequal}, the secure capacity equals
	\begin{align}
\label{eq:capDiamNonun}
	C = \max\limits_{\begin{array}{c} f_i \geq 0, \forall i \\  \sum\limits_{i} f_i = 1 \end{array}} \left[ \sum\limits_{i=1}^{N} f_i C_i - \max\limits_{\begin{array}{c} S \subseteq [N] \\ |S| = K \end{array}} \sum\limits_{i \in S} f_i C_i \right].
	\end{align}
\end{thm}

\begin{proof}
  \textbf{Achievability:}  It is clear {that} we can transmit $\sum\limits_{i=1}^N f_i C_i$ symbols from the source to the destination, by using for  a {fraction} $f_i$  of time the path with capacity $C_i$. 
Thus, each of the  $N$ outgoing edges from the source (and similarly each of the $N$ incoming edges to the destination) will carry $f_1C_1, \ f_2C_2, \ \ldots \ f_NC_N$ packets, respectively. 
The adversary, {in the worst case wiretaps} $K$ edges, which carry the maximum number of packets. Using a similar encryption scheme as we {designed in Section~\ref{sec:unit_edge}, ensures} a secure rate  $\left[ \sum\limits_{i=1}^{N} f_i C_i - \max\limits_S \sum\limits_{i \in S} f_i C_i \right]$, where $S \subseteq [N], \  |S| = K $. By optimizing over the $f_i$'s we get {that $C$ in~\eqref{eq:capDiamNonun} is achievable.}
%  
%  
%  
%  {\color{blue}By optimizing over the $f_i$'s (using linear program below) we get the required result in the theorem. 
%  \begin{align*}
%  \begin{array}{ccc}
%  \max\limits_f & \sum\limits_{i=1}^{N} f_i C_i - Q  & \\
%  \text{s.t.} & Q \geq \sum\limits_{j=1}^{K} f_{i_j} C_{i_j} & \forall \{i_1,i_2,\ldots, i_k\} \subseteq [N] \\
%  & \sum\limits_{i=1}^{N} f_i = 1 & \\
%  & f_i \geq 0 & \forall i \in [N].
%  \end{array}
%  \end{align*}
%  
%  }
\\ 
	\textbf{Outer Bound:} 
{Since $M=1$, at any time instant,} the source can transmit on {at most} one of {its $N$} outgoing {edges.} 
{We} let $\{T^t_{e_{i_t}}, \ t \in [n] \}$ be the symbols transmitted over $n$ such instances, where $e_{i_t}$ denotes the edge used in the {$t$-th} instance. 
Some of these symbols will {flow through} $e_1$, some {through} $e_2$, and similarly some {through} $e_N$, where $e_i$ is  the edge of capacity $C_i$ {outgoing} from the source. Let $T_{e_i}$ {denote} the symbols transmitted on $e_i$ in all such instances. Thus, $\{T^t_{e_{i_t}}, \ t \in [n] \} = \{T_{e_i}, i \in [N] \}$.
	Let $|T_{e_i}| = n_i, \ i \in [N]$ such that $\sum\limits_i n_i = n$.
	Because of the edge capacity constraints we have $H(T_{e_i}) \leq n_i C_i, \ \forall i \in [N]$.	
	Now, by Fano's inequality,
	\begin{align*}
	n C & \leq H(W) = H(W) - H(W|\{T^t_{e_{i_t}}, \ t \in [n] \}) \\
%	& = I (W; \{T^t_{e_{i_t}}, \ t \in [n] \}) = I (W; \{T_{e_i}, i \in [N] \}) \\
	& =I (W;\{T_{e_i}, i \in S \})\!+\!I (W; \{T_{e_i}, i \notin S \} | \{T_{e_i}, i \in S \}) \\ 	
	& \stackrel{(a)}{\leq} \epsilon + \min\limits_{S \subseteq [N], |S| = K} \!\!I (W; \{T_{e_i}, i \notin S \} | \{T_{e_i}, i \in S \}) \\
	& =  \epsilon + \min\limits_{S \subseteq [N], |S| = K}\!\! H (\{T_{e_i}, i \notin S \} | \{T_{e_i}, i \in S \}) \\
	& \leq \epsilon +  \min\limits_{\begin{array}{c}S \subseteq [N] \\ |S| = K \end{array}} \sum\limits_{i \notin S} n_i C_i \\
	& = \epsilon + \sum\limits_{i \in [N]} n_i C_i - \max\limits_{\begin{array}{c}S \subseteq [N] \\ |S| = K \end{array}} \sum\limits_{i \in S} n_i C_i \\
	C & \leq  \frac{\sum\limits_{i \in [N]} n_i C_i - \max\limits_{\begin{array}{c}S \subseteq [N] \\ |S| = K \end{array}} \sum\limits_{i \in S} n_i C_i }{\sum\limits_{i \in [N]} n_i} \\
	& {\implies} C  \leq \sum\limits_{i \in [N]} f_i C_i - \max\limits_{\begin{array}{c}S \subseteq [N] \\ |S| = K \end{array}} \sum\limits_{i \in S} f_i C_i,
	\end{align*}
	where $(a)$ follows from {the} security condition and the choice {of $S$ to have the} tightest bound, and $f_i = {\frac{n_i}{\sum_{i \in [N]} n_i}} \geq 0, {\sum_{i\in[N]} f_i} = 1$. Optimizing over all such choices of $n_i, \ i \in [N]$, we get {that $C$ in~\eqref{eq:capDiamNonun} is an outer bound on the secure capacity. This concludes the proof of Theorem~\ref{thm:3}.}
%	\begin{align*}
%	C \leq \max\limits_{\begin{array}{c} f_i \geq 0, \ \forall i \\ \sum_{i =1}^N  f_i = 1 \end{array}}  \left(\sum\limits_{i \in [N]} f_i C_i - \max\limits_{\begin{array}{c}S \subseteq [N] \\ |S| = K \end{array}} \sum\limits_{i \in S} f_i C_i \right).
%	\end{align*}	
\end{proof}
%The following example illustrates the achievable scheme.

 \paragraph*{Example 3} Consider a diamond network with $N=4$, and $C_1=3$, $C_2= 2$, $C_3= 2$ and $C_4 = 1$ and assume $K=1$. If we were to use each {path the same} number of times, we would get  a secure rate of $\frac{5}{4}$.
 In constrast, the optimal scheme from Theorem~\ref{thm:3} uses the first path {twice,}  the second and third three times each, and does not use the last {path,} achieving a secure rate of $\frac{3}{2}$. Thus, we see that {different from} the traditional network, here we might need {to discard} some of the resources.

\section{Conclusions}
\label{sec:Concl}
{We explored} security over 1-2-1 networks {where, since we need to use beamforming and align beams to activate links,} we {cannot} use all the underlying graph links {simultaneously}, but instead {we} can use each {link} for a fraction of time that we can decide. 
Over such networks, we have shown that we can achieve a {secure capacity that in some cases can} be very close to the {unsecure capacity;} we have derived upper and lower bounds for arbitrary unit capacity networks, and exact capacity characterizations for some special classes of networks.
\bibliographystyle{IEEEtran}
\bibliography{isit2018.bib}

\begin{figure*}
	\appendix
\section*{Proof of Lemma~\ref{lem:infotheory}}

	\infotheory*

	\begin{IEEEproof}
		Assume for all choices of $S \subset {[L],} |S| = m$, $ H(\{X_i, \ i \in S^c\}| \{X_i, \ i \in S\}) > \frac{L-m}{L}  H(\{X_i, \ i \in [L]\})$. Then,	
		\begin{align*}
		{L \choose m} H(\{X_i, \ i \in [L]\})& \stackrel{(a)}{=}  \sum\limits_{\begin{array}{c} S \subset [L] \\ |S| = m \end{array}} \left( H(\{X_i, \ i \in S\}) + H(\{X_i, \ i \in S^c\}| \{X_i, \ i \in S\}) \right) \\
		& \stackrel{(b)}{\geq} 	\sum\limits_{\begin{array}{c} S \subset [L] \\ |S| = m \end{array}} \left( \left(\sum \limits_{i \in S} H(X_i | \{X_j, \ j < i\}) \right) + H(\{X_i, \ i \in S^c\}| \{X_i, \ i \in S\}) \right)\\
		& \stackrel{(c)}{=} {L-1 \choose m-1}  \left(\sum \limits_{i \in [L]} H(X_i | \{X_j, \ j < i\}) \right) + 	\sum\limits_{\begin{array}{c} S \subset [L] \\ |S| = m \end{array}} H(\{X_i, \ i \in S^c\}| \{X_i, \ i \in S\}) \\
		& \stackrel{(d)}{=} {L-1 \choose m-1}  H(\{X_i, \ i \in [L]\})  + 	\sum\limits_{\begin{array}{c} S \subset [L] \\ |S| = m \end{array}} H(\{X_i, \ i \in S^c\}| \{X_i, \ i \in S\}) \\
		& \stackrel{(e)}{>}  {L-1 \choose m-1}  H(\{X_i, \ i \in [L]\})  + 	\sum\limits_{\begin{array}{c} S \subset [L] \\ |S| = m \end{array}} \frac{L-m}{L}  H(\{X_i, \ i \in [L]\}) \\
		& = {L-1 \choose m-1}  H(\{X_i, \ i \in [L]\})  + 	{L \choose m} \frac{L-m}{L}  H(\{X_i, \ i \in [L]\}) \\
		& = {L \choose m} \left( \frac{m}{L} H(\{X_i, \ i \in [L]\}) + \frac{L-m}{L} H(\{X_i, \ i \in [L]\})  \right) \\
		& = {L \choose m} H(\{X_i, \ i \in [L]\}),
		\end{align*}
		and {hence} we get a contradiction. Here $(a)$ is because there are ${L \choose m}$ ways of breaking $\{X_i, \ i \in [L]\}$ into two sets of size $m$ and $L-m$, and then it follows from the chain rule of entropy; $(b)$ follows because for any $S \subset [L]$, we can order $\{X_i, \ i \in S\}$ according to their index, and then we use the chain rule of entropy followed by the condition reduces entropy principle; $(c)$ follows because for each $i \in [L]$, there will be ${L-1 \choose m-1}$ choices of $S$ where this $i$ will be part of $S$; $(d)$ follows again from the chain rule of entropy; and $(e)$ follows because of the assumption in the proof that for all choices of $S \subset [n], |S| = m$, $ H(\{X_i, \ i \in S^c\}| \{X_i, \ i \in S\}) > \frac{L-m}{L}  H(\{X_i, \ i \in [L]\})$.      	
	\end{IEEEproof}
\end{figure*}

\end{document}